\documentclass[11pt]{article}

\usepackage[utf8]{inputenc}
\usepackage{amsmath, amssymb, amsthm,verbatim, graphicx, bbm}
\usepackage[top=1in,bottom=1in,left=1in,right=1in]{geometry}
\usepackage{caption}
\usepackage{subcaption}
\usepackage{color}

\newtheorem{theo}{Theorem} 
\newtheorem{lemma}[theo]{Lemma}
\newtheorem{corol}[theo]{Corollary}

\newtheorem{defn}[theo]{Definition}
\newtheorem{prop}[theo]{Proposition}

\newcommand{\C}{\mathcal{C}}
\newcommand{\E}{\mathcal{E}}
\newcommand{\G}{\mathcal{G}}
\newcommand{\Q}{\mathcal{Q}}
\newcommand{\f}{\mathbb{F}}

\title{Quantum Expander Codes}

\author{Anthony Leverrier\footnote{Inria, EPI SECRET, B.P. 105, 78153 Le Chesnay Cedex, France. Email: \texttt{anthony.leverrier@inria.fr}.},\quad Jean-Pierre Tillich\footnote{Inria, EPI SECRET, B.P. 105, 78153 Le Chesnay Cedex, France. Email: \texttt{jean-pierre.tillich@inria.fr}.},\quad Gilles Z\'emor\footnote{Mathematical Institute, Bordeaux University, France. Email: \texttt{zemor@math.u-bordeaux.fr}.}}

\date{\today}

\begin{document}

\begin{titlepage}

\maketitle
\thispagestyle{empty}

\begin{abstract}
We present an efficient decoding algorithm for constant rate quantum hypergraph-product LDPC codes which provably corrects adversarial errors of weight $\Omega(\sqrt{n})$ for codes of length $n$.
The algorithm runs in time linear in the number of qubits, which makes
its performance the strongest to date for linear-time decoding
of quantum codes. The algorithm relies on expanding properties, not
of the quantum code's factor graph directly, but of the factor graph
of the original classical code it is constructed from.
\end{abstract}

\end{titlepage}

\section{Introduction}
A quantum CSS code is a particular instance of a quantum stabilizer code, and can be defined by two classical binary linear codes $\C_X$ and $\C_Z$ in the ambient space $\f_2^n$, with the property that $\C_X^\perp\subset \C_Z$ and $\C_Z^\perp\subset \C_X$. In other words, the classical codes $\C_X$ and $\C_Z$ come together with respective parity-check matrices $H_X$ and $H_Z$ such that the linear space  $R_X=\C_X^\perp$ generated by the rows of $H_X$ is orthogonal to the row space $R_Z=\C_Z^\perp$ of $H_Z$, where orthogonality is with respect to the standard inner product. An {\em error pattern}
is defined as a couple $(e_X,e_Z)$, where $e_X$ and $e_Z$ are
both binary vectors. The decoder is given the pair of {\em syndromes} $\sigma_X=H_Xe_X^T$ and $\sigma_Z=H_Ze_Z^T$ and decoding succeeds if
it outputs, not necessarily the initial error pattern $(e_X,e_Z)$, but
a couple of the form $(e_X+f_X,e_Z+f_Z)$ where $f_X\in R_Z$ and $f_Z\in R_X$. 
See \cite{got1997} for the equivalence with the stabilizer formalism
and a detailed introduction to quantum coding.

If efficient quantum computing is to be achieved, it will come with a
strong error-correcting component, that will involve very fast decoders,
probably in not much more than linear time in the blocklength $n$. The likeliest candidates for this task are quantum LDPC codes: in the CSS case, an LDPC code is simply a code whose above parity-check matrices $H_X$ and $H_Z$ have row and column weights bounded from above by a constant. Among
recent developments,
the recent paper \cite{got} has shown how fault tolerant quantum
computation with constant multiple overhead can  be obtained,
and quantum LDPC codes are an essential component of the scheme, making
them possibly even more appealing. 

It is natural to hope that the success of classical LDPC codes, both in terms of performance and of decoding efficiency, can eventually be matched in the quantum setting. 
This agenda involves two major difficulties, however. The first one is that coming up with intrinsically good constructions of quantum LDPC codes is in itself a challenge. In particular the random constructions that can be so effective in the classical case do not work at all in the quantum case. Indeed if
one chooses randomly a sparse parity-check matrix $H_X$, then, precisely
since this gives a good classical code, there are no low-weight codewords
in the dual of the row-space of $H_X$ and therefore an appropriate matrix
$H_Z$ does not exist. A testament to this difficulty is given in the
introduction of \cite{TZ14} by way of a list of proposed constructions
of quantum LDPC codes from within the classical coding community that all yield constant minimum distances.
Presently, the known constructions of families of quantum LDPC codes that come with constant rates and minimum distances that grow with the qubit length can be reduced to essentially three constructions. The first consists of quantum codes based on tilings of two-dimensional hyperbolic manifolds (surfaces) that generalize Kitaev's toric code and originate in \cite{FML}. The minimum distance of these codes grows as $\log n$, where $n$ is the qubit length. A recent generalisation of this approach to $4$-dimensional hyperbolic geometry \cite{GL14} yields minimum distances that behave as $n^\epsilon$ where $\epsilon$ is larger than some unknown constant and not more than~$0.3$. Finally, the construction \cite{TZ14} yields codes of constant rate
with minimum distances that grow as
$n^{1/2}$. These codes are perhaps the closest to classical LDPC codes in spirit, since they are constructed by taking a properly defined product of a classical LDPC code with itself. We note that presently all known constructions of quantum codes, even if they are allowed to have
vanishing rate, fail to significantly break the $n^{1/2}$ barrier for the minimum distance and it is an intriguing open question as to whether there exist asymptotically good quantum LDPC codes (i.e. with constant rate and minimum distance linear in $n$).
We also make the side remark that Gottesman \cite{got} requires, for the purpose of fault-tolerant quantum computation, constant rate LDPC codes with good minimum distance properties that should behave well under some sort of adversarial error setting.

The second difficulty in attempting to match the achievements of classical LDPC coding, is to devise efficient decoding algorithms.
The vast majority of decoding algorithms developed for classical LDPC codes rely on iterative  techniques whose ultimate goal is to take decisions on individual bits. In the quantum setting, taking a
decision on an individual qubit is mostly meaningless: 
this is because a decoder who would try to recover the error vector exactly is doomed to fail, since it would be fooled by any error that spans only half a
stabilizer vector (a row vector $f_X$ of $H_X$ or $f_Z$ of $H_Z$) \cite{PC08a}.
The eventual error vector that one must look to output is therefore
only defined up to addition of stabilizer vectors, so that there is no ``correct value'' for a single qubit of
$e_X$ or $e_Z$ which can always be just as well $0$ or $1$. 

Decoding quantum LDPC codes requires therefore additional elements to the classical tool kit. Surface codes mentioned above come with efficient decoding algorithms, however they are very particular to the surface code structure,
which implies that the associated classical codes $\C_X$ and $\C_Z$ are cycle codes of graphs:
full decoding, which is NP-hard in general for linear codes, can be achieved for cycle codes of graphs in polynomial time (with the help of Edmonds' weighted matching algorithm \cite{edmonds}), and this strategy
(which does not really qualify as a local technique) yields a decoding scheme for the quantum code
that achieves vanishing error-probability for random errors. Unfortunately, this technique does not extend to other classes of LDPC codes, and in an adversarial setting is limited to correcting at most
$\log n$ errors, since the minimum distance of surface codes of constant rate can never surpass a logarithm of the
qubit length \cite{delfosse}. Very recently, an alternative decoding algorithm was proposed \cite{hastings2014decoding} 
for the 4-dimensional hyperbolic codes of Guth and Lubotzky that is devised to work in a probabilistic setting and for which its adversarial performance is unclear. The third class
of constant rate quantum codes with growing minimum distance, namely the codes \cite{TZ14}, had
no known decoding algorithm to go with it until the present paper whose prime objective is
to tackle this very problem.

In this paper we devise a decoding algorithm for the product codes \cite{TZ14} that runs in linear time and decodes an arbitrary pattern of errors of any weight up to a constant fraction of the minimum distance, i.e. $cn^{1/2}$ for some constant $c>0$. Our inspiration is the decoding algorithm
of Sipser and Spielman \cite{SS96} which applies to classical LDPC codes whose Tanner graph is an expander graph: recall that ``the'' Tanner graph (there may actually be several for the same code) is a bipartite
graph defined on the set of rows and the set of columns of a parity-check matrix for the code and which
puts an edge between row $i$ and column $j$ if the matrix contains a ``1'' in position $(i,j)$.
The quantum codes under consideration here are products (in a sense to be given precisely below)
of a classical LDPC code $\C$ with itself, and we take the original code $\C$ to be an expander code. The resulting
Tanner graphs of the two classical codes $\C_X$ and $\C_Z$ that make up the quantum code are not strictly speaking expander graphs, but they retain enough expanding structure from the original code
for a decoding algorithm to work. Arguably, this is the first time that an import from classical
LDPC coding theory succeeds in decoding a quantum LDPC code from a non-constant number of errors in an adversarial setting. There are some twists to the original Sipser-Spielman decoding scheme however,
since it guesses values of individual bits and we have pointed out that this strategy cannot carry through to the quantum setting. The solution is to work with {\em generators} rather than qubits:
the generators are the row vectors of $H_X$ and $H_Z$ (thus called because they correspond to generators of the stabilizer group of the code). At each iteration, the decoding algorithm looks for a pattern
of qubits inside the support of a single generator that will decrease the syndrome weight.

Our results also have some significance in the area of local 
testability.
Locally Testable Codes (LTC) play a fundamental role in complexity theory: they have the property that code membership can be verified by querying only a few bits of a word \cite{goldreich2010short}.
More precisely, the number of constraints not satisfied by a word should be proportional to the distance of the word from the code. 
Given their importance, it is natural to ask whether a quantum version of LTC exists, and to investigate their consequences for  the burgeoning field of quantum Hamiltonian complexity, which studies quantum satisfaction problems \cite{gharibian2014quantum}.

Quantum LT codes were recently defined in \cite{aharonov2013quantum}, and these hypothetical objects are mainly characterized by their \emph{soundness} (or \emph{robustness}) $R(\delta)$, i.e. the probability that a word at relative distance $\delta$ from the code violates a randomly chosen constraint. 

While we do not exhibit such quantum LT codes here, we construct codes which are robust for errors of reasonably low weight, up to a constant fraction of the minimum distance: see Corollary \ref{robust} below. Reaching beyond the regime of low weight errors appears to be much harder since it is well-known that the random expander codes at the heart of our construction are not locally testable \cite{ben2005some}.
Interestingly, for our construction, better expansion translates into greater robustness. This should be seen in contrast to results in Ref.~\cite{aharonov2013commuting}, \cite{aharonov2013quantum}, where good expansion (admittedly not of the same graph) appears to hurt the local testability of the quantum codes. We also remark
that in the very recent result of \cite{eldar2015quantum}, quantum
codes are constructed by applying \cite{TZ14} to classical LT codes,
which leads to an alternative form of robustness where errors with small syndrome weight correspond to highly entangled states. 

The remainder of this extended abstract is organized as follows:
Section~\ref{sec:mainresult} describes the code construction with its
expanding properties and states the main result. Section~\ref{basic-algo}
describes the basic decoding algorithm. Section~\ref{sec:analysis}
gives the main points of the analysis of the algorithm, and
states the robustness result. 
Section~\ref{sec:conclusion} gives some concluding comments.
The proofs of the more technical lemmas
are relegated to appendices.

\section{Expansion, code construction and main result}\label{sec:mainresult}

Let $G = (A \cup B,\E)$ be a biregular bipartite graph
with left (resp. right) degree equal to $\Delta_A$ (resp. $\Delta_B$). Let $|A| = n_A$, $|B|=n_B$,
  and suppose $n_B\leq n_A$, so that $\Delta_A\leq\Delta_B$.  We shall write 
$a \sim_G b$ (or more concisely, $a \sim b$ when $G$ is clear from context) to mean that the vertices $a$ and $b$ are adjacent in the graph $G$.
If $S$ is a subset of vertices, denote by $\Gamma(S)$ the
set of all neighbors of vertices of $S$.
Let us say that $G$ is $(\gamma_A,\delta_A)$-{\em left-expanding},
for some constants $\gamma_A, \delta_A>0$,
if for any subset
$S \subseteq A$ with $|S| \leq \gamma_A n_A$ we have $|\Gamma(S)| \geq (1-\delta_A)\Delta_A |S|$.
Similarly, we shall say that $G$ is $(\gamma_B,\delta_B)$-{\em right-expanding}
if for any subset
$S \subseteq B$ with $|S| \leq \gamma_B n_B$ we have $|\Gamma(S)| \geq (1-\delta_B)\Delta_B |S|$.
Finally we shall say that $G$ is $(\gamma_A, \delta_A,\gamma_B, \delta_B)$ left-right expanding
(or simply expanding) is it is both $(\gamma_A,\delta_A)$-{left-expanding} and
$(\gamma_B,\delta_B)$-{right-expanding}.

To any bipartite graph $G$ we may associate the $n_B\times n_A$ matrix $H$,
whose rows are indexed by the vertices of $B$, whose columns are indexed by the vertices of $A$,
and such that $H_{ij}=1$ if $i$ and $j$ are adjacent in $G$ and $H_{ij}=0$ otherwise.
A binary linear code $\C_G$ can be thus defined as the set of vectors $x$ of $\f_2^{n_A}$ such that $Hx^T=0$, i.e. $\C_G$ is the code with parity-check matrix $H$.
Conversely, any code $\C$ has a parity-check matrix $H$ and the binary matrix $H$
can in turn be viewed as an incidence relation between its rows and columns, i.e. a bipartite graph $G$: such a graph
is usually called ``the'' Tanner graph or the factor graph of the code $\C$. A code
has several factor graphs, because it has several parity-check matrices, but it is usually clear which one we are talking about.

Expansion and classical error-correction are connected through the following result of Sipser and Spielman~\cite{SS96}.

\begin{theo}\label{thm:SS}
  Let $G=(A\cup B,\E)$ be a $(\Delta_A,\Delta_B)$-biregular $(\gamma_A,\delta_A)$-left-expanding
  graph. Letting $\Delta_A$ and $\Delta_B$ be fixed and allowing $n_A$ to grow,
there exists a decoding algorithm for the associated code $\C_G$ that runs in time linear in the code length $n_A=|A|$, and that, under the condition $\delta_A<1/4$, corrects any pattern of at most $\frac 12\gamma_An_A$ errors.
\end{theo}

Our objective is to derive a quantum analogue of Theorem~\ref{thm:SS}. The codes we shall
work with are the codes of \cite{TZ14}, whose construction we briefly recall. As mentioned
in the introduction, a CSS code is defined by two classical codes $\C_X$ and $\C_Z$: in our case,
both these classical codes are constructed from a fixed bipartite graph $G=(A\cup B,\E)$.
Let us describe $\C_X$ and $\C_Z$ through their factor graphs $\G_X$ and $\G_Z$. 

The bipartite graph $\G_X$ has left set of vertices $A^2\cup B^2$, and its right set of vertices is $A\times B$. The bipartite graph $\G_Z$ has the same left vertices but its set of
right vertices is $B\times A$. We will find it convenient to denote vertices of 
$\G_X$ and $\G_Z$ by pairs of letters, omitting parentheses to lighten notation, and
to use Greek letters for right vertices of $\G_X$, Latin letters for right vertices of
$\G_Z$, and denote elements of $A^2$ by a Greek letter followed by a Latin letter, and
elements of $B^2$ by a Latin letter followed by a Greek letter. Typical elements of
$A^2,B^2,A\times B,$ and  $B\times A$ will therefore be written respectively, $\alpha a$, $b\beta$, $\alpha\beta$, and $ba$. The incidence structure of the two graphs is defined as follows:
\begin{itemize}
\item In $\G_X$: the set of neighbors of left vertex $\alpha a\in A^2$ is defined as 
      $$\Gamma(\alpha a)=\{\alpha\beta\in A\times B, a\sim_G \beta\}.$$ 
      The set of neighbors of
      left vertex $b\beta\in B^2$ is defined as
      $$\Gamma(b\beta)=\{\alpha\beta\in A\times B, \alpha\sim_G b\}.$$
\item In $\G_Z$: the set of neighbors of left vertex $\alpha a\in A^2$ is defined as 
      $$\Gamma(\alpha a)=\{ba\in B\times A, \alpha\sim_G b\}.$$ 
      The set of neighbors of
      left vertex $b\beta\in B^2$ is defined as
      $$\Gamma(b\beta)=\{ba\in B\times A, a\sim_G \beta\}.$$
\end{itemize}
The corresponding parity-check matrices $H_X$ and $H_Z$ of $\C_X$ and $\C_Z$ may therefore
be written in concise form as:
\begin{align}
H_X &= \left(\mathbbm{1}_{n_A} \otimes H, H^T \otimes \mathbbm{1}_{n_B} \right),\\
H_Z &= \left(H \otimes \mathbbm{1}_{n_A}, \mathbbm{1}_{n_B} \otimes H^T \right).
\end{align}

It is not difficult to check \cite{TZ14} that the rows of $H_X$ are orthogonal to the rows of $H_Z$, so that $(\C_X,\C_Z)$ makes up a valid CSS quantum code. We also see that
the parity-check matrices $H_X$ and $H_Z$ are low density, with constant row weight $\Delta_A+\Delta_B$.
It is proved furthermore
in \cite{TZ14} that the parameters of the quantum code $\Q_G=(\C_X,\C_Z)$ are
\begin{equation}
  \label{eq:parameters}
  [[n = n_A^2+n_B^2, k\geq (n_A-n_B)^2, \min(d,d^T)]]
\end{equation}
where $d$ denotes the {\em minimum distance} of the classical code
$\C_G$, i.e. the minimum number of columns of the $n_B\times n_A$ matrix $H$ that sum to zero, and $d^T$ stands for the associated
{\em transpose minimum distance}, which is the minimum number of {\em rows} of $H$ that sum to zero. This is with the convention that each
of these respective
minima is set to equal $\infty$ if there are no subsets of columns, or rows respectively, that sum to zero. The transpose minimum distance 
associated to a parity-check matrix of a classical linear code is not
a standard parameter because it is usually $\infty$, but this is not
necessarily the case for a matrix associated to a typical bipartite
biregular graph, and it cannot be totally overlooked.

Crucially, if one fixes the degrees $\Delta_A,\Delta_B$ with
$\Delta_A<\Delta_B$ and one takes an infinite family of graphs
$G=(A\cup B,\E)$ with increasing number of vertices $n_A,n_B$,
the rate of the quantum code $\Q_G$ is bounded from below by a non-zero constant and its typical minimum distance
scales as the square-root of its length $n$.

We remark that the construction of \cite{TZ14} is somewhat more general, using two
base graphs $G_1$ and $G_2$ rather than a single graph $G$. The above construction corresponds
to the case when $G=G_1=G_2$ and we choose to restrict ourselves to this setting
for ease of description and notation.

We can now state our main result:

\begin{theo}\label{thm:main}
 Let $G=(A\cup B,\E)$ be a $(\Delta_A,\Delta_B)$-biregular $(\gamma_A,\delta_A,\gamma_B,\delta_B)$-left-right-expanding
  graph. Assume the conditions $\delta_A<1/6$ and $\delta_B<1/6$.  
Letting $\Delta_A$ and $\Delta_B$ be fixed and allowing $n_A,n_B$ to grow,
there exists a decoding algorithm for the associated quantum code $\Q_G$ that runs in time linear in the code length $n=n_A^2+n_B^2$, and that decodes any quantum error pattern of weight less than
\begin{equation}
  \label{eq:n0}
  w_0=\frac{1}{3(1+\Delta_B)} \min \left(  \gamma_A n_A,  \gamma_B n_B\right).
\end{equation}
\end{theo}

\paragraph{{\bf Comments.}}
Theorem~\ref{thm:main} implies in particular that the quantum
code $\Q_G$ must have a minimum distance proportional to
$\min \left(  \gamma_A n_A,  \gamma_B n_B\right)$, and that the
decoding algorithm corrects a number of adversarial errors equal to a constant fraction of the 
minimum distance. We stress that we need both left and right expansion from the graph $G$, and that for these values of $\delta_A$ and $\delta_B$, there are no known constructions that achieve it. However,
the graph $G$ with the required expanding properties may be obtained by random choice with classical
probabilistic arguments, as developed in \cite{Bas81} or \cite{SS96}.

\section{Description of the decoding algorithm}
\label{basic-algo}

We first need to describe the decoding problem precisely. 
An error pattern is a set of coordinates $E_X\subset\{1,\ldots ,n\}$
on which there is an $X$ Pauli error, together with a set of coordinates $E_Z$ on which there is a $Z$ Pauli error. There may be coordinates in
$E_X\cap E_Z$ on which an $X$ error and a $Z$ error occur simultaneously, which is equivalent to a $Y$ Pauli error.
The error pattern can therefore be given in equivalent form as the couple
of binary vectors $e=(e_X,e_Z)$ where $e_X$ and $e_Z$ have supports $E_X$ and $E_Z$ respectively.
The input to the decoder is the {\em syndrome $\sigma(e)$} of the error $e$.
It is made up of the classical syndromes of the vectors $e_X$ and $e_Z$ for the matrices $H_X$ and $H_Z$, i.e. $\sigma(e)=(\sigma_X(e_X),\sigma_Z(e_Z))$, with $\sigma_X(e_X)=H_Xe_X^T$ and $\sigma_Z(e_Z)=H_Ze_Z^T$. 
We will say that {\em the error pattern $e=(e_X,e_Z)$ is correctly decoded} if on input $\sigma(e)$ the decoder outputs
$(e_X+f_X,e_Z+f_Z)$, with $f_X$ and $f_Z$ in the row space of $H_Z$ and the row space of $H_X$ respectively.
Our purpose is to exhibit a decoder that will always correctly decode every error $e$ of weight smaller than the quantity given in Theorem~\ref{thm:main}. The {\em weight} $w(e)$ of the error $e=(e_X,e_Z)$
is the number of coordinates~$i$ in which at least one of the two vectors $e_X,e_Z$ is non-zero. 

Before describing the decoder,
let us fix some additional notation.
The rows of the matrix $H_Z$ are called {\em $Z$-generators} 
(for the purpose of this paper's internal notational coherence: note that
this terminology is not standard in the stabilizer formalism).
Recall that a row of $H_Z$ is indexed by an element $ba$ of $B\times A$.
We choose to identify the corresponding generator with its support in
$A^2\cup B^2$, and denote it by $g_{ba}$, so that we have:
\begin{align}
g_{ba} := \left\{ \alpha a \: : \: \alpha \sim b\right\} \cup \left\{ b \beta \: : \: \beta \sim a\right\}. 
\end{align}

Similarly the $X$-generators are denoted by $g_{\alpha \beta}$ and we write
\begin{align}
g_{\alpha \beta} := \left\{ a \beta \: : \: a \sim \beta \right\} \cup \left\{ \alpha b \: : \: b \sim \alpha\right\}.
\end{align}

The decoding algorithm treats $X$ and $Z$ errors independently.
This means that the decoder applies two separate decoding
procedures, one
consisting of outputting $e_X+f_X$ from $\sigma_X(e_X)$, and the other
outputting $e_Z+f_Z$ from $\sigma_Z(e_Z)$.

We remark that it is enough to prove Theorem~\ref{thm:main}
for error patterns of the form $(e_X,0)$ and of the form $(0,e_Z)$.
We will therefore describe the decoding algorithm for errors
of type $(e_X,0)$, the other case being symmetric. Note that
the weight of the error pattern $w(e_X,0)$ is simply the Hamming
weight of the vector $e_X$, and we denote it by $|e_X|$.

The decoding algorithm for $e_X$
works by going through the generators $g_{ba}$ and for each one, by
looking at whether flipping any pattern of bits strictly decreases the weight of the syndrome. 
More precisely: the algorithm takes as input a syndrome vector
$s\in\f_2^{A\times B}$. It then looks for a vector $e_1\in\f_2^{A^2\cup B^2}$ such that:
\begin{itemize}
\item the support of $e_1$ is included in some generator $g_{ba}$,
\item $s_1=s+\sigma_X(e_1)$, with $|s_1|<|s|$,
\item $\frac{|s|-|s_1|}{|e_1|}$ is maximal, subject to the above two conditions. 
\end{itemize} 
If the algorithm cannot find a vector $e_1$ satisfying the first two conditions, it outputs a decoding failure. 
After $i$ decoding iterations, the algorithm is in possession of
a syndrome vector $s_i\in\f_2^{A\times B}$, together with a vector $e_1+e_2+\cdots +e_i\in\f_2^{A^2\cup B^2}$,
such that
\begin{itemize}
\item the support of every $e_j$, $j=1\ldots i$, is included in some generator,
\item $s_i=s+\sigma_X(e_1+\cdots +e_i)$ and $|s_i|<|s|$.
\end{itemize}
The $(i+1)$-th decoding iteration consists in finding a vector 
$e_{i+1}\in\f_2^{A^2\cup B^2}$ such that
\begin{itemize}
\item the support of $e_{i+1}$ is included in some generator $g_{ba}$,
\item $s_{i+1}=s_i+\sigma_X(e_{i+1})$, with $|s_{i+1}|<|s_i|$,
\item $\frac{|s_i|-|s_{i+1}|}{|e_{i+1}|}$ is maximal, subject to the above two conditions. 
\end{itemize}

The algorithm proceeds until it reaches some iteration $i$
after which it cannot find any generator $g_{ba}$ which enables it to
decrease the weight of $s_i$. If $|s_i|\neq 0$, then it outputs a decoding failure. Otherwise we have $s_i=0$, and the algorithm outputs
the vector $e_O := e_1+e_2+\cdots +e_i$.

We shall prove that if $s=\sigma_X(e_X)$ with $|e_X|$ sufficiently small,
then the decoding algorithm never outputs a failure and its output
$e_O$ satisfies $e_O+e_X\in\C_Z^\perp$, equivalently $e_O+e_X$ is a sum of $Z$-generators, meaning the algorithm has correctly decoded the error pattern $e_X$.

By carefully updating the list of generators to be re-examined after every iteration, we obtain in a classical way an algorithm that runs in time linear in the number $n$ of qubits, in the uniform cost model.

\section{Analysis of the decoding algorithm}\label{sec:analysis}
We first recall some tools and results from \cite{SS96}.

Consider the graph $G=(A\cup B,\E)$ and a subset of vertices $S$. 
Let us define the set $\Gamma_u(S)$ of {\em unique neighbors} of $S$, that is, the set of vertices $v$ such that $v$ has degree 1 in the graph $S\cup\Gamma(S)$ induced by $S$. 
The complement of $\Gamma_u(S)$ in $\Gamma(S)$ will be called the set of {\em multiple neighbors} 
of $S$ and denoted $\Gamma_m(S)$.
Provided that the expansion in the graph $G$ is large enough,  then the graph displays unique-neighbor expansion, in the following sense:

\begin{lemma} \cite{SS96}.
\label{unique-exp}
Let $G = (A \cup B,\E)$ be a $(\gamma_A,\delta_A)$-left-expanding graph
with $\delta_{A} < 1/2$. Then, for any subset
 $S_{A} \subseteq A$ with $|S_{A}|\leq \gamma_{A} n_{A}$, we have:
\begin{align}
|\Gamma_u(S_{A})| \geq (1-2\delta_{A}) \Delta_{A} |S_{A}|.
\end{align}
\end{lemma}
\begin{proof}
Consider the bipartite graph induced by $S_A$ and $\Gamma(S_A)$. 
The number of edges incident to $S_A$ should coincide with the number of edges
incident to $\Gamma(S_A)$,
that is $\Delta_A |S_A| = |\Gamma_u(S_A)| + \tilde{\Delta} |\Gamma_m(S_A)|$ where $\tilde{\Delta} \geq 2$ is the average right degree on $\Gamma_m(S_A)$.
This implies that $|\Gamma_m(S_A)| \leq \frac{1}{2}(\Delta_A |S_A| - |\Gamma_u(S_A)|$. 
Moreover, the expansion in $G$ requires that $|\Gamma_u(S_A)|  \geq (1-\delta_A)\Delta_A |S_A|- |\Gamma_m(S_A)|$. 
Combining both inequalities gives the unique-neighbor expansion.
\end{proof}

We also recall from \cite{SS96}:
\begin{prop}\label{prop:d}
  If $G$ is $(\gamma_A,\delta_A)$-left-expanding for $\delta_A<1/2$, then
  the minimum distance $d$ of the associated classical code $\C_G$ is at least $\gamma_An_A$.
\end{prop}

From Proposition~\ref{prop:d} we immediately obtain:
\begin{corol}\label{cor:dmin}
  If the graph $G=(A\cup B,\E)$ is $(\gamma_A,\delta_A,\gamma_B,\delta_B)$-left-right-expanding with $\delta_A,\delta_B<1/2$, then the minimum distance of the associated quantum code $\Q_G$ is at least $\min(\gamma_An_A,\gamma_Bn_B)$.
\end{corol}
\begin{proof}
  Proposition~\ref{prop:d} implies that the minimum distance of the
  associated classical code is at least $\gamma_An_A$, and, by inverting the roles of $A$ and $B$, that the
  transpose minimum distance is at least $\gamma_Bn_B$. The result follows from \eqref{eq:parameters}. 
\end{proof}

The crucial part of the original decoding strategy of Sipser and Spielman
consists in showing that, for errors of sufficiently small weight,
there must exist (at least) one critical
variable vertex of $A$ in error that has many unique
neighbors. Flipping the bit associated with this vertex decreases the
syndrome value and this eventually leads to showing that one can always
decode correctly by flipping bits that decrease the syndrome weight, provided
the initial error is sufficiently small.
In the present setting, we need a corresponding notion of criticality
for generators. It will build upon the unique neighbor idea with a number
of twists.

\begin{defn}\label{def:critical}
  Let $E\subset A^2\cup B^2$ be a subset of vertices that can be thought of as an error pattern.
  Let us say that a generator $g_{ba} = \left\{ \alpha a \: : \: \alpha \sim b\right\} \cup \left\{ b \beta \: : \: \beta \sim a\right\}$ is {\em critical} (with respect to $E$)
 if it can be partitioned as follows:
\begin{align}\label{eq:gba}
g_{ba} = x_a \cup \bar{x}_a \cup \chi_a \cup x_b \cup \bar{x}_b \cup \chi_b
\end{align}
where
\begin{itemize}
\item $x_a \cup \bar{x}_a \cup \chi_a$ is a partition of $g_{ba}\cap A^2$ and
      $x_b \cup \bar{x}_b \cup \chi_b$ is a partition of $g_{ba}\cap B^2$,
\item $x_a\cup x_b\subset E$, $\bar{x}_a\cap E=\emptyset$ and $\bar{x}_b\cap E=\emptyset$,
\item every vertex in $\Gamma(x_a)\cap\Gamma(x_b)$ has exactly two neighbors in $E$,
\item every vertex in $\Gamma(\bar{x}_a)\cap\Gamma(\bar{x}_b)$ has no neighbor in $E$,
\item every vertex in $\Gamma(x_a)\cap\Gamma(\bar{x}_b)$ and every vertex in $\Gamma(\bar{x}_a)\cap \Gamma(x_b)$ has exactly one neighbor in $E$,
\item $x_a \cup x_b \ne \emptyset$, $|\chi_a| \leq 2 \delta_B \Delta_B, |\chi_b| \leq 2 \delta_A \Delta_A$. 
\end{itemize}
\end{defn}

Note that Definition~\ref{def:critical} implies in particular that 
the syndrome vector has value $0$ in the coordinates indexed by
$\Gamma(x_a) \cap \Gamma(x_b)$ and $\Gamma(\bar{x}_a)\cap\Gamma(\bar{x}_b)$, and
has value $1$ in the coordinates indexed by $\Gamma(x_a) \cap \Gamma(\bar{x}_b)$ and in $\Gamma(\bar{x}_a)\cap \Gamma(x_b)$. See Fig.~\ref{syndrome} for details.

Lemma~\ref{goodgen} below asserts that a critical generator with respect
to an error pattern always exists, whenever the
error pattern $E$ has sufficiently small weight (cardinality), and Lemma~\ref{errordecreases}
claims that its is always possible to modify the pattern of qubits inside a critical generator
in a way as to simultaneously decrease the syndrome weight and not increase the error weight. Recall that we are decoding an $X$-error $e_X$ that we simply denote $e$ from now on.

\begin{lemma}
\label{goodgen}
If the weight of the error $e$ satisfies $0 < |e| \leq \min (\gamma_A n_A, \gamma_B n_B)$, then there exists a critical generator $g_{ba}$ with respect to
the support $E$ of $e$.
\end{lemma}

The proof of Lemma~\ref{goodgen} is given in Appendix~\ref{sec:goodgen}.

If $e$ is a vector of $\f_2^n$, let us call its {\em reduced weight}, 
denoted $w_R(e)$, the smallest Hamming weight of an element of the coset
$e+\C_Z^\perp$, i.e. the set of vectors of the form $e+f$, where $f$ is a sum
of $Z$-generators.

\begin{lemma}
\label{errordecreases}
If the reduced
weight of the error $e$ satisfies  $0<w_R(e)\leq\min (\gamma_A n_A, \gamma_B n_B)$, then there exists a critical generator 
together with a vector $e_1$ whose
support is included in the generator, such that $|\sigma_X(e)|-|\sigma_X(e+e_1)|\geq |e_1|/3$.
In words, flipping the $k$ bits of the support of $e_1$ decreases the syndrome weight by at least $k/3$.
Moreover,  $w_R(e+e_1)\leq w_R(e)$.
\end{lemma}

The proof of Lemma~\ref{errordecreases} is given in Appendix~\ref{sec:errordecreases}.

The decoding algorithm is analyzed in two steps. First we show that it never outputs a decoding failure,
i.e. always decreases the syndrome weight to zero, and secondly we prove the
output error vector $e_O$ is equivalent to the original vector $e$ modulo
the space of $Z$-generators.

Lemma \ref{errordecreases} implies the 
\textit{robustness} of the quantum code. 
\begin{corol}[Robustness]
\label{robust} 
Any error $e$ with reduced weight $w_R(e) < \min (\gamma_A n_A, \gamma_B n_B)$ has a syndrome with weight bounded from below as $|\sigma(e)| \geq \frac{1}{3} w_R(e)$.
\end{corol}

\begin{proof}
Without loss of generality, suppose $e$ is a representative of $e+\C_Z^\perp$
with minimum weight.
As long as the weight of the syndrome is positive, Lemma \ref{errordecreases} guarantees the existence of a generator and a vector $e_1$ whose support is included in the generator such that $|\sigma_X(e)|-|\sigma_X(e+e_1)|\geq |e_1|/3$.
Moreover, since $w_R(e+e_1)\leq |e|$, we can apply Lemma \ref{errordecreases}
again to $e+e_1$ and iterate, say $i$ times, until the syndrome weight reaches 0. 
After $i$ iterations, we obtain that the syndrome of $e+e_1+e_2+\cdots +e_i$
is 0, hence that 
\begin{equation}
  \label{eq:1/6}
  |\sigma_X(e)|\geq \frac 13(|e_1|+|e_2|+\cdots |e_i|) \geq \frac 13 |e_1+\cdots e_i|,
\end{equation}
and that $w_R(e+e_1+\cdots +e_i)\leq |e| < \min (\gamma_A n_A, \gamma_B n_B)$.
This last fact implies by Corollary~\ref{cor:dmin} that $e+e_1+\cdots +e_i$ is
in $C_Z^\perp$, i.e. that $e$ is equal to $e_1+\cdots +e_i$ modulo a sum of
$Z$-generators. 
Inequality \eqref{eq:1/6} proves therefore that $|\sigma_X(e)|\geq \frac 13 w_R(e)$.
\end{proof}

Unfortunately, the decoding algorithm is not guaranteed to follow the good decoding path exhibited by Lemma~\ref{errordecreases}. Indeed, the decoding algorithm simply tries to optimize the weight of the syndrome, but the error weight might increase in the process. Nevertheless, we have:
\begin{lemma}
\label{terminates}
If the Hamming weight of the initial error $e$ is less than $w_0$, then the decoding algorithm never outputs a decoding failure and always correctly decodes $e$.
\end{lemma}
\begin{proof}
The decoding algorithm chooses a sequence of vectors
 $e_1,e_2,\ldots,$ such that $$|\sigma_X(e)|,|\sigma_X(e+e_1)|,\ldots,|\sigma_X(e+e_1+\cdots +e_i)|,\ldots$$ is a decreasing sequence. Set $\varepsilon_0=e,\varepsilon_1=e+e_1,\ldots ,\varepsilon_i=\varepsilon_{i-1}+e_i$.
Now whenever 
\begin{equation}
  \label{eq:epsilon}
  |\varepsilon_i|< \min(\gamma_An_A,\gamma_Bn_B),
\end{equation}
Lemma~\ref{errordecreases} ensures that
\begin{equation}
\frac{|\sigma_X(\varepsilon_i)|-|\sigma_X(\varepsilon_{i+1})|}{|e_{i+1}|} \geq \frac{1}{3}.
\end{equation}
By hypothesis \eqref{eq:epsilon} holds for $i=0$. Suppose it holds for $0,1,\ldots ,i$, we then have
\begin{eqnarray*}
|e_1| & \leq & 3 (|\sigma_X(\varepsilon_0)|-|\sigma_X(\varepsilon_1)|)\\
\dots &\leq & \dots\\
|e_{i+1}| & \leq & 3(|\sigma_X(\varepsilon_i)|-|\sigma_X(\varepsilon_{i+1})|)
\end{eqnarray*}
By summing all these inequalities we obtain 
$ |e_1|+ \dots + |e_{i+1}| \leq 3\left(|\sigma_X(\varepsilon_0)|-|\sigma_X(\varepsilon_{i+1})|\right)
\leq 3|\sigma_X(e)|$. Writing $|\varepsilon_{i+1}|\leq |e| + |e_1|+ \dots + |e_{i+1}|$ we get $|\varepsilon_{i+1}|\leq |e| + 3|\sigma_X(e)|\leq |e| + 3\Delta_B|e|$ since the syndrome $\sigma_X(e)$ has weight at most $\Delta_B|e|$.
From the hypothesis on the Hamming weight of $e$ we get
that \eqref{eq:epsilon} holds for $i+1$, and inductively until the eventual output
of the algorithm $e_O=e_1+\cdots +e_j$, such that $\varepsilon_j = e+e_O$ has
zero syndrome. Condition \eqref{eq:epsilon} translates to
$|e+e_O|<\min(\gamma_An_A,\gamma_Bn_B)$ which ensures 
by Corollary~\ref{cor:dmin} that
$e+e_O\in\C_Z^\perp$, which means exactly that $e_O$ is a valid
representative of the initial error~$e$.
\end{proof}

\section{Concluding comments and questions}\label{sec:conclusion}

We have exhibited a linear-time decoding algorithm that corrects up
to $\Omega(n^{1/2})$ {\em adversarial} quantum errors over $n$ 
qubits. While this is the largest such asymptotic quantity to date,
one would hope to break this barrier and eventually achieve correction
of $\Omega(n)$ errors. If one were to do this with quantum LDPC codes,
this would imply obtaining the elusive proof of existence of low-density codes with a minimum distance scaling linearly in the number of qubits.

Kovalev and Pryadko have shown \cite{KP13} that the codes of \cite{TZ14}
have the potential to correct  number of {\em random} 
depolarizing errors that scales linearly in $n$, with a vanishing probability of decoding error. This is without decoding complexity 
limitations however, and a natural question is whether the ideas of the present
paper can extend to decoding $\Omega(n)$ random errors in linear
or quasi-linear time.

We have worked to achieve the smallest possible value of
$\delta_A,\delta_B$ in Theorem~\ref{thm:main}, i.e. the smallest possible
expansion coefficient for the base graph $G$. Can the bound
$\delta_A,\delta_B < 1/6$ be decreased further~? A somewhat related
question is whether the left-right expanding base graph $G$ can be
obtained constructively, rather than by random choice~?


\pagebreak
\appendix

\begin{center}
  {\LARGE \bf Appendices}
\end{center}

\begin{figure}[h]
\centering
 \includegraphics[width=.6\linewidth]{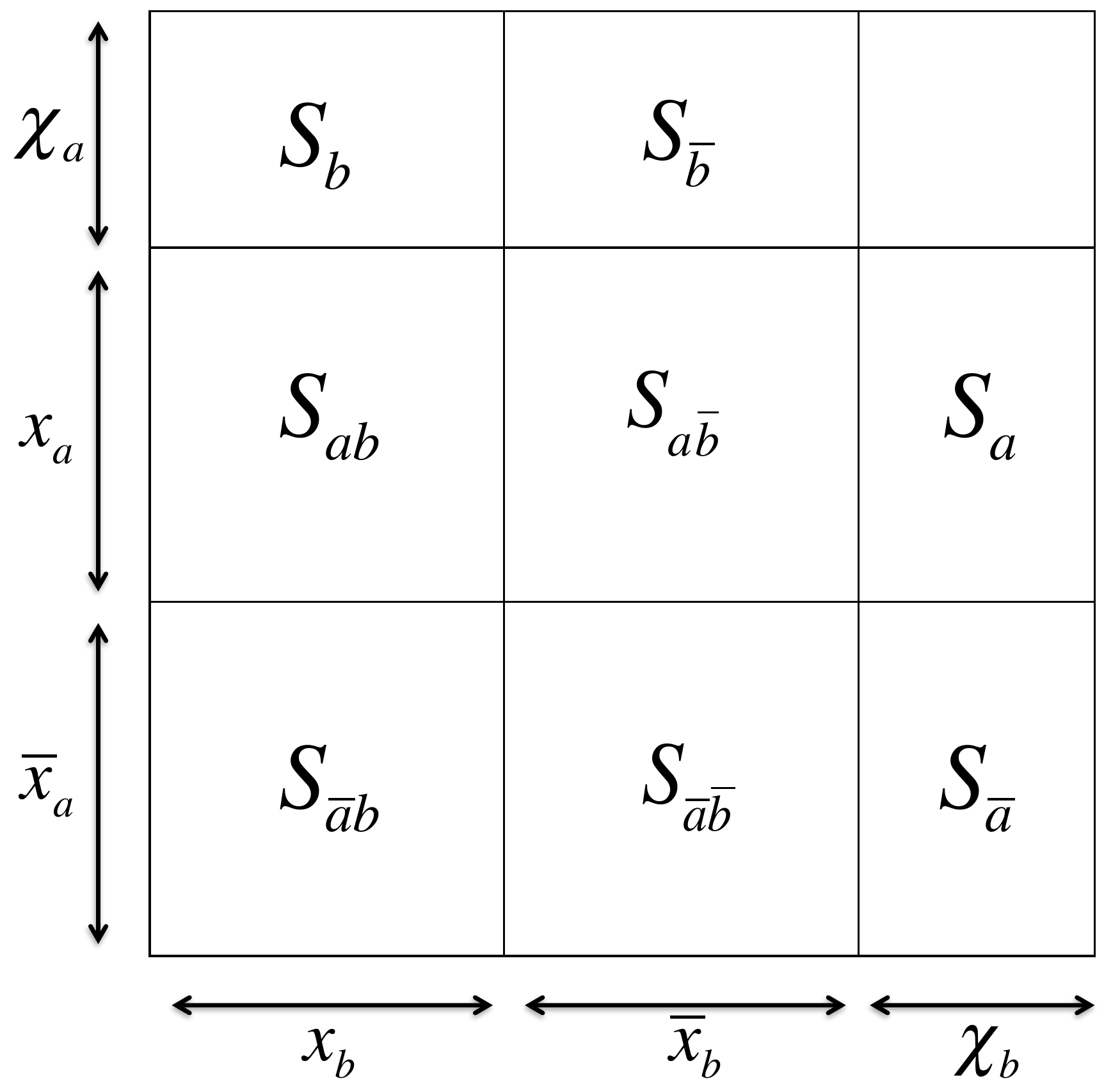}
\caption{A critical generator $g_{ba}$. The generator $g_{ba}$ is partitioned into 6 sets: $x_a$ and $x_b$ contain the errors we want to correct, $\bar{x}_a$ and $\bar{y}_b$ do not contain any error, and $\chi_a$ and $\chi_b$ have a small relative size. When flipping the qubits corresponding to $x_a$ and $x_b$ (or the qubits in $\bar{x}_a$ and $\bar{x}_b$), the syndromes in $S_{a \bar{b}}$ and $S_{\bar{a}b}$ see their weight decrease from 1 to 0 while the syndromes in the sets $S_{ab}$ and $S_{\bar{a}\bar{b}}$ remain unchanged; finally either the syndrome weights in $S_a \cup S_b$, or in $S_{\bar{a}} \cup S_{\bar{b}}$ can increase by 1 while the others stay unchanged. 
}
  \label{syndrome}
\end{figure}

\section{Proof of Lemma \ref{goodgen}}\label{sec:goodgen}

First, we need the following extra notation. Let $S$ be a subset of $A$
in the graph $G=(A\cup B,\E)$:
For a vertex $a\in S$,
we denote as follows, respectively, the set of unique neighbors of $a$ in the subgraph of $G$ induced by
$S\cup\Gamma(S)$, and the set of multiple neighbors of $a$ in the same induced graph:
\begin{align}\label{eq:um}
\Gamma_u^{S}(a) := \Gamma(a) \cap \Gamma_u(S), \qquad \Gamma_m^{S}(a)  := \Gamma(a) \cap\Gamma_m(S). 
\end{align}

Denote now by $E \ne \emptyset$ the support of the $X$ error $e$ to be corrected and define $E_A := E \cap A^2$ (resp.~$E_B := E\cap B^2$) the set of error vertices in $A^2$ (resp.~in $B^2$). 

Suppose first $E_B = \emptyset$. 
Consider the projection $E_A^2$ of $E_A$ on the second coordinate:
\begin{align}
E_A^2 := \left\{ a \in A \:  : \: A \times \{a\} \cap E_A \ne \emptyset \right\}.
\end{align}
Clearly $|E_A^2| \leq |E_A| \leq w(e) \leq \gamma_A n_A$. 
Lemma \ref{unique-exp} then implies that $|\Gamma_u(E_A^2)| \geq (1-2\delta_A) \Delta_A |E_A^2|$.
Therefore, using the notation \eqref{eq:um},
there exists $a \in E_A^2$ such that $|\Gamma_u^{E_A^2}(a)| \geq (1-2\delta_A) \Delta_A$.
Pick some $\alpha a \in E_A$ and $b \sim \alpha$. 
Then the generator $g_{ba}$ can be partitioned into four components $g_{ba} = x_a \cup \bar{x}_a \cup \bar{x}_b \cup \chi_b$ which satisfy the requirements of the lemma. Specifically,
we define $x_a=g_{ba}\cap E_A$, $\bar{x}_a$ as the complement of $x_a$ in $g_{ba}\cap A^2$,
$$\chi_{b}=\{b\beta\in g_{ba}\cap B^2 \: : \: \beta\in\Gamma_m^{E_A^2}(a)\}$$
and $\bar{x}_b$ as the complement of $\chi_b$ in $g_{ba}\cap B^2$. In words, 
$x_a$ and $\bar{x}_a$ partition the $A^2$ part of the generator $g_{ba}$ into error and error-free components, and $\bar{x}_b$ and $\chi_b$ partition the $B^2$ part into,
respectively, the set
whose second coordinate is a unique neighbor and the set whose second coordinate is a multiple
neighbor with respect to $E_A^2$.

The proof is identical with the roles of $A$ and $B$ interchanged if $E_A = \emptyset$. 

Let us turn our attention to the remaining case where both $E_A\neq\emptyset$ and $E_B \ne \emptyset$. 
We define $E_A^2$ as above and choose a coordinate $a \in E_A^2$ such that $\left| \Gamma_u^{E_A^2}(a) \right| \geq (1-2\delta_A) \Delta_A$. From now on, we use the shorthand $\Gamma_u(a)$ for this set, leaving $E_A^2$ implicit. 
Next, we define:
$$E_{B,a} := \{b\beta \in E_B, \, \beta \in\Gamma_u(a)\}.$$
If $E_{B,a} = \emptyset$, we proceed as in the case where $E_B = \emptyset$. 
Otherwise, by assumption on the weight of the error, namely $|e| \leq \gamma_B n_B$, Lemma \ref{unique-exp} applies again, this time to the set $E_{B,a}^1$ corresponding to the projection of $E_{B,a}$ on its first coordinate,
\begin{align}
E_{B,a}^1 := \left\{b \: : \: \exists\beta\in\Gamma_u(a), b\beta\in E_B \right\}.
\end{align}
This implies that there exists $b \in E_{B,a}^1$ such that $\left|\Gamma_u^{ E_{B,a}^1}(b) \right| \geq (1-2\delta_B)\Delta_B$. 
Using the shorthand $\Gamma_u(b) := \Gamma_u^{ E_{B,a}^1}(b)$,
we define:
\begin{align}
x_a & :=\{ \alpha a \in E_A \: : \: \alpha \in \Gamma_u(b)\}\\ 
\bar{x}_a & := \{ \alpha a \not\in E_A \: : \: \alpha \in  \Gamma_u(b)\}\\
x_b & :=\{ b \beta \in E_B \: : \: \beta \in \Gamma_u(a)\}\\ 
\bar{x}_b & := \{ b \beta \not\in E_B \: : \: \beta \in  \Gamma_u(a)\}.
\end{align}
The construction ensures that every element $\alpha\beta$ that belongs both to
the neighborhood of $x_a\cup \bar{x}_a$ and of $x_b\cup \bar{x}_b$ is such that
$\alpha$ is a unique neighbor of $b$ {\em and} $\beta$ is simultaneously a unique neighbor of 
$a$, unique neighborhood being understood with respect to $E$.
Moreover $x_b \ne \emptyset$ by construction of the set $E_{B,a}^1 \ne \emptyset$.

\section{Proof of Lemma \ref{errordecreases}}\label{sec:errordecreases}

We consider the generator $g_{ba}$ promised by Lemma \ref{goodgen}.
Without loss of generality, we may suppose that the error vector $e$ is
in reduced form, i.e. it is the vector of lowest Hamming weight in the coset
$e+\C_Z^\perp$. The consequence is that the number of qubit vertices in error
$|E\cap g_{ba}|$ within the generator $g_{ba}$ is not more than $(\Delta_A+\Delta_B)/2$,
otherwise replacing $e$ by $e+e'$, with $e'$ the vector having $g_{ba}$ for support, would yield a vector with strictly smaller Hamming weight within
the coset $e+\C_Z^\perp$. In particular we have:
\begin{equation}
  \label{eq:reduced}
  |x_a\cup x_b|\leq \frac 12(\Delta_A+\Delta_B).
\end{equation}

Let us introduce the following reduced variables:
\begin{align}
x := \frac{|x_a|}{\Delta_B}, \qquad z := \frac{|\chi_a|}{\Delta_B}, \qquad \bar{x} := 1-x-z,\\
y := \frac{|x_b|}{\Delta_A}, \qquad t := \frac{|\chi_b|}{\Delta_A}, \qquad \bar{y} := 1-y-t.
\end{align}
By assumption, $x + y >0$.

The expansion condition $\delta_A, \delta_B < \frac{1}{6}$ implies that $z,t <1/3$ by the last condition of Definition~\ref{def:critical}. 
However, $z \Delta_B$ and $x\Delta_A$ are integers, 
this implies $3z \Delta_B \leq 3 \Delta_B-1$ and 
$3t \Delta_A \leq 3 \Delta_A-1$, which in turn give:
\begin{align}
\label{bound-z-t}
0 \leq z \leq \frac{1}{3} - \frac{1}{3 \Delta_B}, \quad 0 \leq t \leq \frac{1}{3} - \frac{1}{3 \Delta_A}.
\end{align}

We are looking for a vector $e_1$, whose support is included in $g_{ba}$,
such that
\begin{itemize}
\item[i.] the syndrome of $e+e_1$ has strictly smaller weight than the syndrome of $e$, and the difference $|\sigma_X(e)|-|\sigma_X(e+e_1)|$
is at least $|e_1|/3$,
\item[ii.] the reduced weight $w_R(e+e_1)$ is at most equal to the Hamming weight $|e|$ of $e$.
\end{itemize}

We will consider four cases:
\begin{enumerate}
\item if $x+y\leq 2/3$, then the vector $e_1$ is chosen to have support
$x_a \cup x_b$,
\item if $x \leq \bar{x}$ and $y \leq \bar{y}$, then the vector
$e_1$ is chosen to have support $x_a \cup x_b$,
\item if the above two hypotheses do not hold and
if $x > \bar{x}$ and $y > \bar{y}$, then the vector
$e_1$ is chosen to have support either $\bar{x}_a \cup \bar{x}_b$,
or its complement in $g_{ba}$,
\item in the remaining cases, we show that either
 $x_a \cup x_b$ or  $\bar{x}_a \cup \bar{x}_b$, or
 $x_a \cup x_b\cup\chi_a\cup\chi_b$ is an adequate choice for
 the support of $e_1$.
\end{enumerate}

Let us introduce the partition of $\Gamma(g_{ba})$ induced by the partition~\eqref{eq:gba}:
$$\Gamma(g_{ba})=S_a\cup S_b\cup S_{\bar{a}} \cup S_{\bar{b}} \cup S_{ab}\cup
S_{a\bar{b}}\cup S_{\bar{a}b} \cup S_{\bar{a}\bar{b}}\cup \bar{S}$$
with 
$$
\begin{array}{llllllll}
  S_a\!\! &= \Gamma(x_a)\cap\Gamma(\chi_b)& S_b\!\! &= \Gamma(x_b)\cap\Gamma(\chi_a)&
  S_{\bar{a}}\!\! &= \Gamma(\bar{x}_a)\cap\Gamma(\chi_b)&
  S_{\bar{b}}\!\! &= \Gamma(\bar{x}_b)\cap\Gamma(\chi_a)\\
  S_{ab}\!\! &= \Gamma(x_a)\cap\Gamma(x_b)& S_{a\bar{b}}\!\!&=\Gamma(x_a)\cap\Gamma(\bar{x}_b)&
S_{\bar{a}b}\!\!&=\Gamma(\bar{x}_a)\cap\Gamma(x_b)&
S_{\bar{a}\bar{b}}\!\!&=\Gamma(\bar{x}_a)\cap\Gamma(\bar{x}_b)\\
\bar{S}\!\!&=\Gamma(\chi_a)\cap\Gamma(\chi_b)&&&&&&
\end{array}
$$
as represented on Figure~\ref{syndrome}. 

Let us denote by $\partial$ the decrease of the syndrome weight 
when we flip $x_a \cup x_b$ (i.e. choose $e_1$ to have support $x_a \cup x_b$). Similarly we denote by $\bar{\partial}$ the decrease
of the syndrome weight when we flip either $\bar{x}_a \cup \bar{x}_b$
or $x_a \cup x_b\cup\chi_a\cup\chi_b$.

These quantities satisfy:
\begin{align}
\partial/(\Delta_A \Delta_B) & \geq x \bar{y} + \bar{x} y - x t - yz\label{eq:partial}\\
\bar{\partial}/(\Delta_A \Delta_B) & \geq x \bar{y} + \bar{x} y - \bar{x} t - \bar{y}z.\label{eq:partial2}
\end{align}
This is because when we flip the bit values on the set $x_a\cup x_b$, then the value
of the syndrome is changed from 1 to 0 on the support $S_{a \bar{b}}$ and $S_{\bar{a}b}$, remains at 0 on the supports $S_{ab}$ and $S_{\bar{a}\bar{b}}$,
may possibly change from 0 to 1 on the supports $S_a \cup S_b$, and
remains unchanged in the other regions. Similarly, when we
flip the bit values on the set $\bar{x}_a \cup \bar{x}_b$, the syndrome
is flipped from 1 to 0 on the support $S_{a \bar{b}}$ and $S_{\bar{a}b}$
and possibly from 0 to 1 on $S_{\bar{a}} \cup S_{\bar{b}}$ while the other regions
remain unchanged.

We now address the four cases:
\begin{enumerate}
\item Suppose $x+y\leq 2/3.$ Then \eqref{eq:partial} can be rewritten as:
\begin{align*}
\frac{\partial}{\Delta_A \Delta_B} &
\geq x \bar{y} + \bar{x}y-xt-yz \\
& \geq x(1-y-t)+y(1-x-z)-xt-yz\\
& \geq x+y - 2 xy - 2(tx + yz)\\
& \geq \left\{ \frac{1}{3}(x+y)-2xy \right\} +x \left(\frac{2}{3}-2t \right)+y \left(\frac{2}{3}-2z \right)\\
& \geq \left\{ \frac{1}{3}(x+y)-2xy \right\} + \frac{2}{3} \left\{ \frac{x}{\Delta_A} + \frac{y}{\Delta_B} \right\}
\end{align*}
by applying Eq.\eqref{bound-z-t} to $z$ and $t$. The first term is
nonnegative for $x+y \leq 2/3$, so that we obtain:
$$
\partial \geq \frac{2}{3} |x_a \cup x_b| > \frac 13 |x_a \cup x_b|.
$$
Furthermore, when $e_1$ has support $x_a\cup x_b$, the support
of $e_1$ is included in the support of $e$, hence $|e+e_1|<|e|$ and
$w_R(e+e_1)<|e|$.
\item If $x \leq \bar{x}$ and $y \leq \bar{y}$, then 
$y=1-t-\bar{y}\leq \bar{y}$ gives $\bar{y}\geq \frac 12 - \frac t2$,
hence, applying \eqref{bound-z-t},
$$\bar{y}-t \geq \frac{1}{2\Delta_A}.$$
Similarly we have
$\bar{x}-z\geq \frac{1}{2\Delta_B},$
and \eqref{eq:partial} gives
\begin{align*}
\partial & \geq \Delta_A \Delta_B \left[ x (\bar{y}-t) + y(\bar{x} -z)\right]\\
& \geq \frac 12 \left[ x \Delta_B + y \Delta_A \right]\\
& \geq \frac 12 \left| x_a \cup x_b \right|.
\end{align*}
The set $x_a\cup x_b$ is again an adequate choice for the support of $e_1$, and $w_R(e+e_1)<|e|$ as before.
\item If $x > \bar{x}$ and $y > \bar{y}$, then by an argument symmetrical to the preceding case, exchanging $x$ with $\bar{x}$ and
$y$ with $\bar{y}$, we get:
\begin{align*}
\bar{\partial} \geq \frac 12 \left| \bar{x}_a \cup \bar{x}_b \right|.
\end{align*} 
We remark that we must have $|\bar{x}_a \cup \bar{x}_b |>0$, otherwise
$z,t\leq 1/3$ imply that $|x_a\cup x_b|>(\Delta_A+\Delta_B)/2$, which
contradicts our assumption \eqref{eq:reduced}. Therefore, choosing
$e_1$ to have support $\bar{x}_a \cup \bar{x}_b$ gives a strict
decrease in the syndrome weight that is at least $|e_1|/3$.
Now in this case, the {\em Hamming weight} $|e+e_1|$  is larger
than $|e|$. However we need to consider its {\em reduced weight}.
A vector equivalent modulo $C_Z^\perp$ to $e+e_1$ is $e+e_1'$
with the support of $e_1'$ being $x_a \cup x_b\cup\chi_a\cup\chi_b$.
The Hamming weight of $e+e_1'$ is 
  $$|e+e_1'|\leq |e| -(x+y)+(z+t).$$
But since $x+y>2/3$ and $z,t,<1/3$, we have $|e+e_1'|< |e|$,
so that $w_R(e+e_1)<|e|$.
\item Finally, in the remaining case we consider
 $\partial + \bar{\partial}$ to show that one of the two quantities,
$\partial$ or $\bar{\partial}$ is sufficiently large.
Adding \eqref{eq:partial} and \eqref{eq:partial2} yields
$$\frac{1}{\Delta_A \Delta_B} (\partial +\bar{\partial})
  \geq 2 x \bar{y} + 2 \bar{x} y - t(1-z) - z(1-t).
$$
We now observe that the negation of the condition
$$(x \leq \bar{x}\; \text{and}\; y \leq \bar{y})\quad\text{or}\quad
   (x > \bar{x}\; \text{and}\; y > \bar{y})$$
implies $2 x \bar{y} + 2 \bar{x} y\geq (x+\bar{x})(y+\bar{y})=(1-z)(1-t)$, from which we get:
 $$\frac{1}{\Delta_A \Delta_B} (\partial +\bar{\partial})
  \geq (1-z)(1-t) -  t(1-z) - z(1-t).$$
The formal equality
  $$(1-z)(1-t) -  z(1-z) - t(1-t) = 3 \left(\frac{2}{3} -z\right)\left( \frac{2}{3}-t\right) -\frac{1}{3}$$
yields therefore, when combined with \eqref{bound-z-t},
$$\frac{1}{\Delta_A \Delta_B} (\partial +\bar{\partial})
  \geq \frac{1}{3} \left(1 + \frac{1}{\Delta_A}\right)\left(1 + \frac{1}{\Delta_B} \right)-\frac{1}{3}$$
Hence, after rearranging,
\begin{align*}
\max \left\{ \partial,\bar{\partial} \right\}\geq \frac 12(\partial+\bar{\partial}) \geq \frac{1}{3} \frac{\Delta_{A} + \Delta_B}{2}.
\end{align*}
\begin{itemize}
\item if $\max \left\{ \partial,\bar{\partial}\right\}=\partial$,
then we choose $e_1$ to have support $x_a\cup x_b$, and
we get $\partial\geq \frac 13|e_1|$ from condition \eqref{eq:reduced}. We have $w_R(e+e_1)\leq |e|$ as in cases 1. and 2.
\item if $\max \left\{ \partial,\bar{\partial}\right\}=\bar{\partial}$,
then 
\begin{itemize}
\item either $|\bar{x}_a \cup \bar{x}_b|\leq \frac{\Delta_{A} + \Delta_B}{2}$, in which case we set $e_1$ to have support
$\bar{x}_a \cup \bar{x}_b$ and have  $\bar{\partial}\geq \frac 13|e_1|$
and $w_R(e+e_1)\leq |e|$ by the same argument as in case 3,
\item or $|\bar{x}_a \cup \bar{x}_b|> \frac{\Delta_{A} + \Delta_B}{2}$,
in which case we set $e_1$ to have support $x_a\cup x_b\cup\chi_a\cup\chi_b$ so as to again have $\bar{\partial}\geq \frac 13|e_1|$. The Hamming weight $|e+e_1|$ is at most $|e|-(x+y)+(z+t)$,
which is less than $|e|$ as in case 3, so that again $w_R(e+e_1)<|e|$.
\end{itemize}
\end{itemize}
\end{enumerate}

\end{document}